\documentclass[10pt|11pt|12pt]{article}
\usepackage{cite}
\usepackage{amssymb,amsmath,latexsym,mathrsfs}
\usepackage{psfrag}
\usepackage{bm}
\usepackage{cite}
\usepackage{url}
\usepackage{color}
\usepackage{graphicx}
\usepackage{caption}
\usepackage{subcaption}
\usepackage{booktabs}
\usepackage{multirow} 
\usepackage{mathtools} 
\usepackage{extarrows} 
\usepackage{centernot}

\usepackage{authblk}

\usepackage{cite}
\usepackage{amsthm}

\newtheorem{theorem}{Theorem}[section]

\newtheorem{lm}[theorem]{Lemma}

\newtheorem{compsol*}{Complete Solution}

\usepackage{datetime}

\newdate{date}{30}{05}{2020}
\date{\displaydate{date}}

\begin{document}

\title{\LARGE \bf  Planning of  School Teaching
during COVID-19.}

\author{ Alberto Gandolfi }

\affil{NYU Abu Dhabi}

\vspace{5mm}

%
%
%

\maketitle

\thispagestyle{plain}
\pagestyle{plain}

\let\thefootnote\relax\footnotetext{Address: New York University Abu Dhabi, 
PO Box 129188, Abu Dhabi, UAE}
\let\thefootnote\relax\footnotetext{
AMS 2010 subject classifications. .

Key words and phrases. 

Running Title: High School planning}

\begin{abstract}
More than one billion students are out of school
because of Covid-19, forced to a remote learning
that has several drawbacks and has been hurriedly 
arranged; in addition, most countries
are currently uncertain on how to  plan  school activities for the
2020-2021 school year; all of this makes learning and education some of the biggest world issues
of the current pandemic. Unfortunately, due
to the length of the incubation period of Covid-19, 
full opening of schools seems to be impractical 
till a vaccine is available.
In order to support the possibility of  some  in-person learning,
we study a mathematical model of the 
diffusion of the epidemic due to  school opening,
and evaluate plans aimed at containing the extra
Covid-19 cases due to school activities while
ensuring an adequate number of in-class learning periods.

We consider a SEAIR model with an external source of 
infection and  
a suitable loss function; after  a realistic parameter selection,
we numerically determine optimal school opening
strategies by simulated annealing.
It turns out that blended models, with almost periodic alternations of in-class
and remote teaching days or weeks, are generally optimal. 
Besides containing Covid-19 diffusion, these solutions 
could be pedagogically acceptable, and could
also become a driving model for the society at large.

In a prototypical example, the optimal strategy results in the school opening $90$ days out 
of $200$ with the number of Covid-19 cases among the individuals related to the school increasing by 
 about $66\%$, instead of the about $250\%$ increase that
would have been a consequence of  full opening.

\end{abstract}

\newpage

\section{Introduction}

More than one billion students are out of school
because of Covid-19 \cite{Unesco 2020-03-04};
 most of them are now using 
  remote learning practices 
that present several drawbacks  \cite{Unesco 2020-03-10}  and
in most cases  have been hurriedly 
arranged.

 In addition, in most countries there is uncertainty on
 how to  plan  school activities for the
2020-2021 school year \cite{Unicef 20}. As a result,  learning and
education are two of the biggest world issues
of the current pandemic \cite{Ciano20, Lyst 20}. Unfortunately, due
to the length of the incubation period of Covid-19 of 
about $5$ days \cite{QL, K-W}, 
full opening of schools seems to be impractical 
till a vaccine is available; this has  been experienced by the
first experiments of reopening \cite{Nelson et al. 20, CBS 20},
and it is further
 confirmed by the simulations in the present paper.

In order to support the possibility of  some  in-person learning,
we study a mathematical model of the 
diffusion of the epidemic due to  school opening.
More specifically, we consider an adaptation of the SEIR model
to the problem of planning school teaching 
during the school year 2020-2021, when the COVID-19
epidemic is very likely to  still be active.
The aim is to optimize the number of
in-class teaching days vs. remote ones in a  school setting,
while containing the number of extra
infectious transmissions due to opening
of the school.

To this purpose, we set-up a suitable 
system of ordinary differential equations
describing the evolution of the epidemic
among the personnel and the students of a school.
The model, in analogy to SEIR, has susceptible individuals 
(S)
to whom the virus can be transmitted
either from external sources (when
they are not at school), or from contacts with other
individuals
within the school. Susceptible individuals  become pre-symptomatic or, equivalently, exposed
(E)
at transmission, entering a latency
period  in which they are already
contagious; they then either show symptoms (I), or
remain asymptomatic (A); both situations resolve, 
and individuals recover (R). As numbers are
limited, and COVID-19 mortality is particularly
low among young individuals, we disregard mortality.
The key differences with the usual SEIR model \cite{CHBC-C} are:
an external source
of infection \cite{Naji2014}, \cite{BANERJEE2014}, the possibility of transmission limited
to $7$ hours per working day, a control allowing
preventative closure of the school each day, and
the presence of asymptomatic individuals. The
asymptomatic members of the community present
a challenge, as they cannot be easily detected
while still being infectious for a relatively
long time: we assume
 accessibility to some type of fast screening  
 with a minimal
sensitivity \cite{Wyllie2020}.

We further assume that the objective of the planner is to 
maximize the number of days with in-class 
vs. remote teaching, and simultaneously minimize
the number of extra infectious transmissions due
the opening of the school. For this, it is 
essential that the planner determines
 the relative weight $\tau$ of these two
objectives, but our set-up
applies to any possible value of $\tau$.
 Once a value of $\tau$ is fixed, the aim
of the planner becomes then to minimize a functional
that combines
 these two objectives
by the relative weight. To achieve her/his targets,  the planner can decide
 ahead of time
which days are in-class and which are
restricted to remote teaching; for simplicity, we consider here a plan for the whole year, but 
monitoring and en route adjustments can be also incorporated. We remain agnostic as to the value
of the relative weight $\tau$, and develop tools
to analyze all possible scenarios.

We start from a careful evaluation
of all the parameters of the model, except
$\tau$, and 
for each scheduling plan we numerically evaluate
the loss functional. We repeat this for 
random selections of the scheduling plans,
obtaining a numerical description of the 
distribution of the outcomes in terms of 
remote days and additional infectious transmissions.
In addition, for each value of $\tau$, 
we use simulated annealing to optimize
with respect to 
 weekly openings or closures.
 A model sensitivity analysis is then carried out
to assess the effect of parameter evaluation errors.
 This gives a complete picture of the 
 potentialities of the optimization procedure.

As we assume an external source of infection,
a fraction of the individuals related to the school
 gets infected even with complete closure
and remote only teaching. In a prototypical simulation,   
$12.2\%$ of the individuals are subject of an infectious transmission
during the school year from contacts outside the school; they become 
pre-symptomatic or, as we indicate them, exposed (E).
Recall that in our model, exposed individuals 
have received an infectious transmission; according
to current estimates \cite{Nishiura2020}, 
about half of the exposed individuals develop symptoms (I). On the other hand, if 
the teaching is completely in-class,
about $42.9\%$ of the individuals get the virus:
that is almost half of the school population, and
 an increment of $251.63\%$ with respect to
 completely remote. 
The outcome of the optimization procedure 
depends on the  relative weight $\tau$,
but most optimal solutions turn out
to be blended, with periodic alternations
between weeks of in-class activities, and
weeks of remote teaching: solutions of this
type have been advocated and are being planned in Scotland \cite{Lyst 20}. In one such solution,
there are $90$ days of in-class teaching,
with the additional infectious transmissions
contained to $7.9\%$, which is an increase of $64.75\%$
of the cases
with respect to remote only activity.
In a school of $1000$ individuals this corresponds
to creating an extra $79$ cases of COVID-19
patients, next to the $122$ who would have been infected
in any case; considering that about half of
the cases are asymptomatic, this corresponds
to having about $40$ extra symptomatic cases
in exchange for almost half of the year
spent in class. The administrators
can decide to realize a stricter or
more relaxed containment of the extra cases by
assigning more or less weight to infectious transmissions.
The solution to achieve the above results is
  (0, 0, 0, 1, 0, 1, 0, 0, 0, 1, 0, 1, 0, 1, 0, 1, 0, 1, 0, 1, 0, 1, 0, 
1, 0, 1, 0, 1, 0, 1, 0, 1, 0, 1, 0, 1, 0, 1, 0, 1),
where $0$ indicates a week of remote, and $1$ a week 
of in-class activities; notice that
there is an initial  period of closure, 
before progressing to a regular alternation of
openings and closures.
If $18$ out of the $40$ weeks are open at random,
the same number of in-class days as above but with no plan,
then the increment in the number of cases is
$88\%$ on average, with a Standard Deviation of 
$10.25\%$: making an random selection of 
the weeks of in-class activities would thus result
in a substantial reduction of the number of
extra infections; on the other hand,
 the optimal solution is statistically better, as it is 
more than $2$ SD's below what the average
random selection would provide, and it is also not subject to 
possible adverse fluctuations.

In Section
\ref{SensitivityAnalysis} we carry out a sensitivity analysis, which shows that
the optimal planning is pretty stable 
in case of some extra opening or remote days,
and is only moderately affected by 
errors in parameter selection, within certain ranges.
The most significant parameters are 
the infectious rate for contacts in the school 
$\beta$, the 
duration of the incubation period before symptoms are
developed $1/\gamma$, and the capability of screening for
asymptomatic individuals $\eta$. Our results are acceptable
if
the incubation period is below $7$ days on average:
while current estimation give about $5$ days \cite{QL, K-W},
this parameter needs to be monitored from the 
evolving studies. In this study, we take the infectious rate for contacts
in the school at a much higher than all current estimates,
but since a school is a densely populated
environment, precautions and distancing methods must 
applied as much as possible not to overcome even our pessimistic 
selection. For the fraction of
detected asymptomatic individuals, we assume 
that this is at least $45\%$: this fraction
relies on testing, that should then be performed
regularly, possibly with simple and inexpensive, even if not
totally reliable, methods \cite{Chekani-Azar2020}.

\medskip

{\bf Related studies.} Optimization methods similar to those used
 in this paper have appeared in many 
other contexts.
In fact, several recent 
studies have considered optimal planning
strategies in response to  general epidemics
\cite{CHBC-C, jones2020optimal} or, more specifically,
to the
COVID-19 pandemic \cite{alvarez2020simple, GKK}; 
 to our knowledge, however,  none has considered the
current optimization problem for school planning.

\medskip

{\bf Limitations.} There are several  limitations to this study.
Our results are only a first indication of a modeling methodology for the search of an optimal trade off between
in-class teaching and containment of infectious transmissions.
Even if parameters are carefully and realistically
selected, the values are based on information
 known at the time of this study; when a parameter
 has a wide range of variability we settle for
 one representative value; in addition,  the worked out examples referred to an ideal school:
for each particular concrete situation, 
one needs to adapt the model to the specific
case. 

We also did not consider other alternatives, such as
having half of the classes, or reducing the 
number of in-class students for each class. These
alternatives could be conveniently incorporated
in a more elaborate model.

\section{A  SEAIR model with 
external source and containment}
\subsection{Epidemic model } \label{epidemic model}

We consider SEAIR, a  version of the SIR model
 \cite{CHBC-C}
with an external source of infection
\cite{Naji2014, Blackwood2018},
and a control.
The population is divided into: susceptible (S), pre-symptomatic or, equivalently, exposed (E),
asymptomatic (A),
infected (I) and recovered (R).
Variables are normalized so that $S+E+A+I+R=1$. 
 
 We assume that susceptible individuals
 might become exposed (E),
a phase  in which they have contracted the virus and are contagious, without showing symptoms. 
The contagion can be caused by contact with
other individuals with viral load, at a rate
$\beta c(t)$ ($c(t)$ is the control, as described below);
or, alternatively, it is caused by an infectious 
contact outside the school taking place at rate $\alpha$.
Exposed individuals either develop 
symptoms at a constant
rate $\delta$, becoming infected, or progress into being asymptomatic 
 with rate $\gamma$.
The viral load is carried by exposed (E), asymptomatic (A) or infectious (I) individuals; however, infectious 
individuals are
assumed to be isolated, while a large fraction of
asymptomatic (A) is supposed to be detected; hence,
if the school is open then
susceptible individuals are assumed to enter in contact
only with exposed, and a fraction $\eta$ of asymptomatic. 

Both asymptomatic and infected individuals recover at  rate $\rho$.

With these assumptions, the system of OdE
modeling infectious transmissions of concern
for the school is:
\begin{align} 
\text{Susceptible:\quad}\frac{dS}{dt} & =-\alpha S-\beta S c(t)(E+\eta A)\label{eq:dS}\\
\text{Exposed:\quad}\frac{dE}{dt} & =\beta S c(t)(E+\eta A)+
\alpha S-(\gamma+\delta)E \label{eq:dE}\\
\text{Asymptomatic:\quad}\frac{dA}{dt} & = \gamma E
-\rho A \label{eq:dA}\\
\text{Infected:\quad}\frac{dI}{dt} & = \delta E-\rho I\label{eq:dI}\\
\text{Recovered:\quad}\frac{dR}{dt} & =\rho (A+I)\label{eq:dR}
\end{align}
for a certain time interval $[0,T]$.
We consider $40$ weeks, and count time in hours, so that
$T=6720$.
The initial population at the beginning of the
school year might be partly immune due to previous
infections, but to avoid questions about efficacy and
duration of the immunity, we assume that the initial
population
consists primarily of susceptible, $S(0) \approx 1$, and a small fraction of exposed, 
so that $S(0)+E(0)=1$. 

The function $c(t)$ describes the control variable, and 
it is $c(t)=0$ for all times when the school is closed; these
include all hours except 
the opening times.
For each of the $200$ working days, 
$c(t)$ can be $0$ again if remote activities have been decided
for that day, or $1$
if teaching is in person. Formally, with $t$
measured in hours, let $\mathcal M$ indicate the class of controls $c(t)$
such that
\begin{equation} \label{c}
    c(t) =
    \begin{cases} c_{i,j} \quad \text{ if } 
    i= \lfloor t/168 \rfloor, j=
    \lfloor (t - 168 \lfloor t/168 \rfloor)/24 \rfloor \leq 5
    \quad \text{ and } \\
    \quad \quad \quad 8 \leq t -168 \lfloor t/168 \rfloor
    -24  \lfloor (t - 168 \lfloor t/168 \rfloor)/24 \rfloor
    \leq 15\\
     0 \quad \text{ otherwise }
    \end{cases}
\end{equation}
for some $c_{i,j} \in \{0,1\}, i=1, \dots, 40,
j=1,\dots,7$,
and $0 \leq t \leq 6720.$

\subsection{Mathematical analysis } \label{Math Anal}
\begin{lm} \label{preserve}
For each solution of \eqref{eq:dS}-\eqref{eq:dR},  the total population
$S+E+A+I+R$ is preserved.
\end{lm}
\begin{proof}
If  $\Vec X(t)=(S(t),E(t),A(t),I(t),R(t))$ is a solution,
then let $\phi=S+E+A+I+R$; we have that $\phi(0)=1$  and 
$\frac{d\phi}{dt}=\frac{d(\phi-1)}{dt}
=-n(\phi-1)$, so that, since $(\phi-1)(0)=0$, necessarily $\phi \equiv 1$ by uniqueness of solutions of linear
differential equations.
\end{proof}
Notice that, in addition, $0\leq S,E,A,I, R
\leq 1$.

\begin{lm}
For all initial conditions $\Vec X(0)=(S(0),E(0),A(0),I(0),R(0))$ such that 
$0\leq S(0),E(0),A(0),I(0), R(0)
\leq 1$ there exists a unique solution
 of the system \eqref{eq:dS}-\eqref{eq:dR}
for all times.

\end{lm}
\begin{proof}  
Since $c(t) \in \{0,1\}$ and 
$S,E,A,I,R$ are
bounded, the r.h.s. of the system \eqref{eq:dS}-\eqref{eq:dR} is Lipschitz \cite{Adkins12} in the variables 
$S, E, A, I , R$: in fact,  let $\Vec{F}(c,\Vec{X})$ be the vector-valued function having as components the right-hand sides of the S-E-A-I-R-D differential equations;
 we can then rewrite the system in vector form
$$
\Vec{X}'=\Vec{F}(c,\Vec{X}),\,\, \Vec{X}(0)=\Vec{X}^0.
$$
We have that $\frac{\partial \vec F_i}{\partial \vec X_j}
\leq c_1 \max_{1\leq i\leq 5}\left\{\max_{[0,T]}|X_i(t)|\right\} \leq c_1$  for some $c_1 >0$. This is a 
 sufficient condition
for existence and uniqueness of solutions
in each interval of continuity of $c(t)$ \cite{Adkins12}.
As  $c(t)$ has only jump discontinuities,
the unique solution of an interval can
be uniquely continued as a continuous function
into the next interval
\cite{Adkins12};
hence there always exist a unique solution for all
times $t$ of the system \eqref{eq:dS}-\eqref{eq:dR}. 
\end{proof}

\begin{lm}
If $(S(0)+E(0)+A(0)+I(0)+R(0))=1$ then there is a unique stationary solution, namely
$S=E=A=I=0, R=1$, which then attracts the solution for all initial conditions.
\end{lm}
\begin{proof}
For a stationary solution, $\Vec{X}'=0$ implies $S=0$ by
\eqref{eq:dS}. Then, \eqref{eq:dE} implies $E=0$,
and \eqref{eq:dA}-\eqref{eq:dI} imply $A=I=0$.
By Lemma \ref{preserve}, it must then be $R=1$.
\end{proof}
We are, however, interested in a finite time interval
$T$, and hence in the transient solution up to time $T$.

\subsection{Planner's objectives } \label{plann obj}

The planner's objectives are incorporated in the model by 
a loss function, defined as follows.
One of the aims of the planner is to contain
the number of extra infectious transmissions due to
the opening of the school; 
let $  S_0(t)$ be the fraction of susceptible 
with completely remote teaching, that is with
$c \equiv 0$, at time $t$;
 $S_0(0)- S_0(T)$ is then the fraction of individual
who receive an infectious transmission even if
the school never opens; then, let $S(t)=S_c(t)$ be the fraction 
of susceptible with opening plan $c(t)$; we have that
$S_0(T)-S(T)=S_0(T)-S_c(T)=(S_c(0)-S_c(T))-(S_0(0)-S_0(T))
$ is the fraction of extra infectious
transmissions due to the opening of the school according
to plan $c(t)$. In addition, let $N(c)=400-\intop_0^Tc(t)dt/7 $
be the number of remote teaching days (the
factor $1/7$ is due to the number of hours
that the school is open in a regular day).
The 
loss function combining the two effects is then
\[ 
\mathcal{L}=(S_0(T)-S(T)) +(N(c)/\tau)
\]
where $\tau$ is the relative weight of 
a day of in-class teaching to  infectious transmissions:
$\tau/100$ can be interpreted as the number of
days of in-class teaching that  the planner considers
equivalent
to a $1\%$ increase in the infectious transmissions
in the school.

The planner's objective becomes then
\begin{equation} \label{LL}
    \min_{c \in \mathcal{M}} \mathcal{L}
\end{equation}
with $\mathcal{M}$ as in \eqref{c}.
It is easy to see that using the system \eqref{eq:dS}-\eqref{eq:dR} the optimization problem
can be cast in a more standard form \cite{Fleming75,GKK}, in which 
$$
\mathcal{L}= E_c(T)-E_0(T)
+ \int_0^T \left(  (\gamma+\delta)(E_c(t)-E_0(t)
+ \frac{1}{7 \tau} c^2(t)\right) dt
$$
where $E_0(t)$ is the exposed component of the 
solution of the system \eqref{eq:dS}-\eqref{eq:dR}
for the case $c\equiv 0$; however, as we optimize over the 
very restricted class of functions \eqref{c}, 
which are piecewise constant and
 depend on the finite number of parameters
$c_{i,j}$'s, the general theory of control
optimization is not needed here: \eqref{LL} becomes a
 discrete optimization problem.

\section{Simulations}

\subsection{Parameter selection}\label{calibration}
We select the epidemic parameters based on current 
observations. As time is counted in hours, we
need to scale all available estimates, usually
expressed in days, by a factor of $24$.

There is a large variability in
the estimations of the COVDI-19 infection rate $\beta$,
with estimated values tipically around $0.25$-$3$ \cite{AT}; 
however, as the school is likely to elicit more
frequent contacts, we adopt a fairly higher value
of $\beta = 0.9/24 \approx 3.75\times 10^{-2}$.

Duration of the latency period after infection
and before symptoms are developed has been 
estimated in about $5$ days (see for example \cite{QL} and \cite{K-W}), so that $\gamma + \delta \approx 1/5 \times 1/24=0.2/24$;
the fraction of asymptomatic is also quite problematic,
with estimates ranging from $5\%$ to $60\%$,
we adopt a value of $50\%$, slightly higher
than the average \cite{Nishiura2020};
therefore, we take $\gamma=\delta =
0.1/24 \approx 4.17  \times 10^{-3}$.
Similarly, the average recovery period is about
$7$ days, for mild cases
\cite{Byrne}, suggesting $\rho = 0.14/24 \approx5.8 \times 10^{-3}$;
 more severe cases (I) are excluded from 
 contacts, so their recovery rate is irrelevant: we
 use the same value of $\rho$.

 The parameter $\eta$ describes how asymptomatic
individuals are excluded from contacts with the other
individuals; such a separation depends 
on the availability of detecting tests: we assume
that a sufficiently reliable test is available, with a 
$90\%$ test sensitivity, so that 
$\eta = 0.1$. This assumption is subject to 
parameter sensitivity analysis in Section \ref{SensitivityAnalysis} below, where we see 
that a test sensitivity  of at least $45\%$ 
is needed for our calculations to make sense.

We finally consider the external rate
of infection.
At this moment, countries, with few
exceptions, have reported at most $1\%$
of infected, but data is not
considered reliable \cite{verity2020estimates}; 
we take an external infection rate 
$\alpha \approx  1.78 \times 10^{-5}$ such that
in the observation period of $40$ weeks
there is a moderate, but non negligible 
number of cases even with the school
being totally closed; with our selection
of parameters the individuals infected outside
of the school will be  about $12.2\%$ of the total.
The parameter $\alpha$ will be carefully monitored
in the sensitivity analysis, but its value
is seen to be irrelevant to our conclusions.

\begin{table}[htbp]
  \centering
  \caption{Recap of the model parameters and their 
estimated values.}
    \begin{tabular}{lr}
         \multicolumn{1}{l}{Parameter} & \multicolumn{1}{l}{Selected value}
         \\ \hline \vspace{0.1cm}
    $\alpha$ & $1.78 \times 10^{-5}$  \\ \vspace{0.1cm}     
    $\beta$ & $3.75 \times10^{-2} $  \\ \vspace{0.1cm}
    $\gamma$ & $4.17  \times10^{-3}$ \\ \vspace{0.1cm}
    $\delta$ & $4.17  \times10^{-3} $ \\  \vspace{0.1cm}
    $\rho$ &$5.83 \times 10^{-3}  $ \\  \vspace{0.1cm}
     $\eta$ & $0.1$   \\  \hline
    \end{tabular}%
  \label{ModelParameters}%
\end{table}%
As initial condition, we assume that there are a few
exposed individuals already present at the start
of the school year, so we start from $E(0)=0.01$


\subsection{Results}\label{Results}
For simplicity, we have assumed here that
the decision about remote or in-class teaching
is taken for each week, so the control $c_i=c_{i,j}$
depends only on $i =1, \dots, 40$. This
still includes $2^{40}$ possible policies.

With the parameters as in Table \ref{ModelParameters}, Column $2$, we have optimized by Simulated Annealing \cite{Kirkpatrick83} 
as follows.
First we have selected  one value of $\tau$; 
we have taken  the range $\tau \in [0.5,2000]$, 
this corresponds to between $0.5$ and $20$
 days of in-class
teaching rated equivalent to a $1\%$ increase
in the number of infected;
for each value of $\tau$, we have selected
a random vector $v \in \{0,1\}^{40}$, and then,
starting from the policy $c_i=v_i$, we have
optimized by Simulated Annealing.
Figure \ref{InfectedClosureDays} 
shows a plot of fraction of infected vs. days
of closure for each step of each simulation.

\begin{figure}[h!]
    \centering
    \includegraphics[scale=0.8]{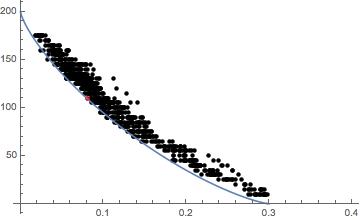}
    \caption{Fraction of infected vs  days
    of remote teaching, several random combinations
    and optimal ones. The optimal once are close to the
    continuous curve. One special optimal solution in red}
    \label{InfectedClosureDays}
\end{figure}
The optimal results for the various values of 
$\tau$ lie approximately on the curve
$(\frac{x}{0.3})^{0.76}+(\frac{y}{200})^{0.76}=1$,
which is also plotted in Figure \ref{InfectedClosureDays}.
Points approximately on the curve correspond to 
the optimal policies for the various values of 
$\tau$. One can see the effect of optimization
with respect to a random selection.
The above curve is only numerically determined,
and its significance is still unclear.

We have then selected one of the optimal
policies, plotted in red in Figure \ref{InfectedClosureDays}.
It corresponds to 
the following selection of open
and closed weeks 
\begin{equation}\label{OptimalPolicy}
\{0, 0, 0, 1, 0, 1, 0, 0, 0, 1, 0, 1, 0, 1, 0, 1, 0, 1, 0, 1, 0, 1, 0, 
1, 0, 1, 0, 1, 0, 1, 0, 1, 0, 1, 0, 1, 0, 1, 0, 1\},
\end{equation}
in which $0$ stands for a week of remote teaching.
Notice that the solution is quite periodic, with
alternated weeks of in-class and remote activities;
Figure \ref{PlotEpidemicOptimalControl}
shows the evolution of the epidemic functions
under the optimal policy \eqref{OptimalPolicy}.
\begin{figure}[h!]
    \centering
    \includegraphics[scale=0.5]{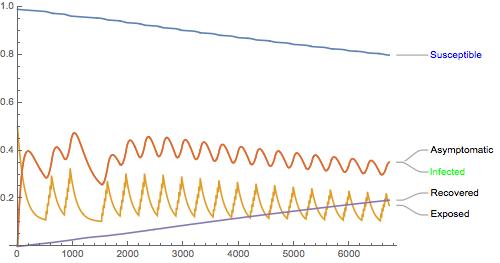}
    \caption{Epidemic functions
under the optimal policy \eqref{OptimalPolicy}. For 
visualization purposes, $E(t)$ is multiplied by $50$, and
$A(t)$ and $I(t)$ are multiplied by $150$. Time is 
in hours.}
    \label{PlotEpidemicOptimalControl}
\end{figure}
Table  \ref{ResultsOptimal} summarizes the 
change in infectious transmissions with the optimal solution;
notice that there is an increase of $64.75\%$ in the number
of the positive cases due to the $90$ days of 
school opening vs. an increase of 
$251\%$ if the school was open all the $200$ days.
This realizes 
 an effective containment of 
  the extra diffusion of the virus, while
  still having the possibility of offering a 
  substantial portion of the teaching in person.

\begin{table}[htbp]
  \centering
  \caption{Performance of optimal and random solutions
  with $90$ days of in-class teaching. First $4$ columns:
  total number of infectious transmissions; last 
  $3$ columns: percentage increase with respect
  to remote only solution.
  }
    \begin{tabular}{lrrrrrrr}
       & \multicolumn{1}{l}{{\tiny Remote}} & \multicolumn{1}{l}{ 
        {\tiny Optimal}} &\multicolumn{1}{l}{{\tiny Random}}& \multicolumn{1}{l}{{\tiny In-class}} &
        \multicolumn{1}{l}{{\tiny Change  w/: Optimal}}&
        \multicolumn{1}{l}{{\tiny Random}}&\multicolumn{1}{l}{{\tiny In-class}} \\ \hline \vspace{0.1cm}
    $\%$ & $ 12.20$ & $20.31$ &$22.96\pm 1.25$ & $42.93$  &
    $64.75$&$88\pm 10.25$ &$251.88$ \\   \hline    \end{tabular}%
  \label{ResultsOptimal}%
\end{table}%

We have also explored the possibility of selecting
at random which weeks to open or close, with the
constraint of having $90$ days of in-class activities.
A simulation of the distribution gives 
then the increment in the number of cases is
$88\%$ on average, with a Standard Deviation of 
$10.25\%$, hence the optimal solution is 
more than $2$ SD's below what the average
random selection would provide.

\vskip 2cm
\newpage

\section{Sensitivity Analysis} \label{SensitivityAnalysis}

We provide a two parts sensitivity analysis for our model.
\bigskip

The first part concerns errors in determination of 
the optimal solution. If one of the weeks 
for which remote teaching is planned by the
optimal solution is instead changed into 
an in-class week, then
the increment in the number of cases
due to school opening ranges from an extra
$70\%$ to an extra $76\%$, compared to the extra
$66\%$ of the optimal solution. If one week 
is moved remote from in-class, then the increment in the number of cases
is reduced to about $62\%$. If the planner had a
value for each in-class day that  lead to
the solution with $90$ in-class days, altering
the number of weeks of opening is not optimal,
but it does not substantially change the final outcome:
therefore, once an optimal policy has been determined, 
extra openings or closures can be adopted if
needed without major changes in the final outcome.
In addition, once the flexible blended model has been adopted,
one can quickly and easily adapt to
 major changes in the external conditions, such 
 as a sudden outbreak or a vaccine.

\bigskip
The second part of the sensitivity analysis
concerns 
 each of the
parameters of the model used to 
determine the optimal solution
\eqref{OptimalPolicy}. For each parameter, we  compare
the increase in the number of cases from the
remote solution, and we contrast this with the 
increase for the complete in-class option.
The results are reported in Table
\ref{SensitivityAnalysis}.

The last column of the table indicates the
acceptable range of parameters:
as the increase is $64.75\%$ if the 
selected parameters are correct (see Table
\ref{ResultsOptimal}), we accept 
up to a two-fold increase in the number of
cases, which is a $100\%$,
and report the resulting range.

\begin{table}[htbp]
  \centering
  \caption{Sensitivity analysis. For each parameter, the second
  column gives the tested range; the third column 
  indicates the change in number of cases compared
  to the change in case of full opening for the
  best value in the range; the fourth column 
  indicates the worst case; the last column
  reports the acceptable range in which 
  the number of cases doubles at the most.
  }
    \begin{tabular}{llrrr}
      parameter & \multicolumn{1}{l}{test range} &  \multicolumn{1}{l}{Minimal effect} & \multicolumn{1}{l}{Maximal effect}& \multicolumn{1}{l}{Accept. range} \\ \hline \vspace{0.1cm}
    $\alpha$ & $0.1 \mbox{-} 15\times 10^{-5} $& $ 16\% \text{ vs. } 32\% \quad$
    & $71\% \text{ vs. } 487\%$
     & $0.1 \mbox{-} 15 \times 10^{-5}$\\ 
     \vspace{0.1cm}
    $\beta$ & $0.1 \mbox{-} 10\times 10^{-2} $&$  0.8\% \text{ vs. } 2\%\quad $&$ 
    391\% \text{ vs. } 664\% $& $0.1 \mbox{-} 5\times 10^{-2} $\\
    \vspace{0.1cm}
    $\gamma (= \delta)$ &$ 0.1 \mbox{-} 10 \times 10^{-3}$&$  23\% \text{ vs. } 64\% \quad$&
    $722\% \text{ vs. } 722\% $&$ 3 \mbox{-} 10 \times 10^{-3}$\\
    \vspace{0.1cm}
    $\rho$ & $0.1 \mbox{-} 10 \times 10^{-3}$& $ 63\% \text{ vs. } 239\% \quad$&
    $224\% \text{ vs. } 630\% $& $1 \mbox{-} 10\times 10^{-3}$ \\
    \vspace{0.1cm}
    $\eta$ & $0 \mbox{-} 1 $&  $60\% \text{ vs. } 221\%\quad $&
    $139\% \text{ vs. } 497\% $&$ 0 \mbox{-} 0.55$ \\
    \hline    \end{tabular}%
  \label{SensitivityAnalysis}%
\end{table}%

We see that changes in the rate of external infections $\alpha$
does not alter the effectiveness of the optimal 
planning, as we always get a substantial reduction from
full opening, and never more than $71\%$ increase; 
the range of $\alpha$ covers the possibilities that
between $1.6\%$
and $64\%$ of the individuals of the school would
be infected outside of the school even with complete closure; this 
covers all possible scenarios of evolution of the
virus.

The internal transmission rate $\beta$ is, on the other
hand, quite significant; if it goes above $5\times 
10^{-2} \times 24 =1.2 $ then the number of cases would more
than double due to the school opening: current
estimates are around $0.3$, but this limitation indicates
that one has to make sure that there are not too many
occasions of infectious transmission within the school.

The incubation period before symptoms are
developed $1/(\gamma+\delta)$ cannot exceed 
$\frac{1}{6 \times 24}10^{3} \approx 7$ days on average, or else the number of cases would more than
double with the optimal school activities: while
this is now estimated to be $5$ days on average,
one needs to monitor whether more accurate studies confirm
this assessment.

The symptomatic period $1/\rho$ also is assumed to
not exceed $\frac{1}{ 24}10^{3} \approx 41$ days
on average, but this is a very safe bound.

Finally, the fraction of undetected asymptomatic
$\eta$ must not exceed $55\%$: we assume
availability of a easy, rapid and cheap
test with a rate of false negative not exceeding
$55\%$, such as CRISPR \cite{Chekani-Azar2020}.

\section{Conclusions}
We have considered the issue of planning in-class activities for 
the school year '20-'21. Education is, in fact,  one of the areas in which the
measures to contain the current Covid-19 pandemic 
have hit the most, with many countries and local authorities
 struggling to find acceptable plans for next year.
 
 To aid such planning, we have set up an optimization problem
aimed at increasing the number of in-class
 vs. remote teaching, while 
containing the number of additional Covid-19 cases
that would be determined by in-class school activities.
The model involves differential equations to simulate
the infectious transmissions, more precisely 
an  SEAIR model with external source
of infection and a control, a loss function
combining the quantities to be minimized,
and a numerical optimization procedure.

Our model has confirmed that 
  the length of the incubation period of Covid-19
  makes it  impractical to have fully in-class activities:
  in a prototypical example, this would lead to an increase
  by more than $250\%$ of the Covid-19 cases
  among the individuals involved with the school.
  
 We have then obtained a numerical expression for the 
curve of optimal strategies parametrized by
the relative importance of in-class days vs.
extra infectious transmissions.
For a typical value of such relative importance, we
have
determined the optimal solution, which appears
to be a blended model,  with alternated weeks
of in-class and remote teaching.
The results turned out to be  rather stable for
possible errors in the estimation of
model parameters, the more critical ones
being the average incubation period, 
the internal transmission rate during school 
opening, and the fraction of detected
asymptomatic individuals.

With the optimal strategy, the increase in the number of cases in the prototypical 
example, is around $66\%$. A random selection of
the weeks of in-class teaching would be much less efficient, 
with an increase of about $88\%$ of the cases, but would
still constitute a sharp reduction from full opening.

We think that this analysis, adapted to specific situations, could
offer a very viable alternative for 
planning the school activities of the 
'20-'21 school year. Some other aspects
of socio-economic life could then 
be arranged around the blended model, in such a way that
most children in the world will be able to enjoy an acceptable
learning experience before a vaccine will hopefully allow the
resumption of the usual school activities.

 \bigskip

Contact address:
NYU Abu Dhabi 
Saadiyat Island
P.O Box 129188
Abu Dhabi, UAE

email: ag189@nyu.edu


\small 
\begin{thebibliography}{99}

\bibitem{Unesco 2020-03-04}[Unesco 2020-03-04] https://plus.google.com/+UNESCO (2020-03-04). "COVID-19 Educational Disruption and Response". UNESCO. Retrieved 2020-05-24.

 \bibitem{Unesco 2020-03-10}[Unesco 2020-03-10]  "Adverse consequences of school closures". UNESCO. 2020-03-10. Retrieved 2020-03-15.
 
 \bibitem{Unicef 20}[Unicef 20] What will a return to school during the COVID-19 pandemic look like?
What parents need to know about school reopening in the age of coronavirus.
Retrieved 2020-06-05,
 
  \bibitem{Ciano20}[Ciano20] R. Cano: Prepping to reopen, California schools desperate for guidance, money. Calmatters,
Retrieved 2020-01-05.

\bibitem{Lyst 20}[Lyst 20] Coronavirus: What is a blended model of learning?
By Catherine Lyst
BBC Scotland, Retrieved 2020-01-05. 

\bibitem{Nelson et al. 20}[Nelson et al. 20]
Soraya Sarhaddi Nelson, Benjamin Restle and Monika Müller-Kroll:
In Brief: COVID-19 outbreak leads city of Göttingen to shut schools through June 7, 
KCRW Berlin, Retrieved 2020-06-05.

\bibitem{CBS 20}[CBS 20] CBSNews: Coronavirus flare-ups force France to re-close some schools
Retrieved 2020-05-18.


\bibitem[Chowell et al. 2009]{CHBC-C}  Chowell, Gerardo and Hyman, James M and Bettencourt, Lu{\'\i}s MA and Castillo-Chavez, Carlos:  Mathematical and statistical estimation approaches in epidemiology.
Springer (2009).

\bibitem[Naji2014]{Naji2014}  R K  Naji, A A Muhseen:  Modeling And Analysis Of An SVIRS Epidemic
Model Involving External Sources Of Disease, INTERNATIONAL JOURNAL OF TECHNOLOGY ENHANCEMENTS AND EMERGING ENGINEERING RESEARCH,
{\bf 2}, 18 (2014), 2347-4289.


\bibitem[Banerjee2014]{BANERJEE2014}  Banerjee, S., Chatterjee, A., AND Shakkottai, S: Epidemic
thresholds with external agents,
IEEE INFOCOM 2014-IEEE Conference on Computer Communications (2014), 2202?2210.

\bibitem[Wyllie2020]{Wyllie2020}  Wyllie, A.L. et al.:  Saliva is more sensitive for SARS-CoV-2 detection in COVID-19 patients than nasopharyngeal swabs.
medRxiv, Cold Spring Harbor Laboratory Press (2020), 2020.04.16.20067835.

\bibitem[Chekani-Azar2020]{Chekani-Azar2020}  Chekani-Azar S, Gharib Mombeni E, Birhan M, and Yousefi M.:
 CRISPR/Cas9 gene editing technology and its application to the coronavirus disease
(COVID-19), a review.
J Life Sci Biomed {\bf 10}, 1 (2020), 01-09.

\bibitem[alvarez2020]{alvarez2020simple} Alvarez, Fernando E and Argente, David and Lippi, Francesco: A simple planning problem for covid-19 lockdown.
National Bureau of Economic Research, working paper 26981 (2020).

\bibitem[Blackwood2018]{Blackwood2018}  Julie C. Blackwood and Lauren M. Childs:  An introduction to
compartmental modeling for the budding infectious disease modeler
Letters in Biomathematics {\bf 5}, 1 (2018), 195-221.

\bibitem[Adkins12]{Adkins12}  Adkins, William A.; Davidson, Mark G.:  Ordinary Differential Equations.
Springer-Verlag, Berlin (2012).

\bibitem[Fleming75]{Fleming75}  W.H. Fleming and R.W. Rishel:  Deterministic and Stochastic Optimal Control.
Springer-Verlag, Berlin (1975).


\bibitem[Grigorieva2020]{GKK}  Grigorieva, Ellina and Khailov, Evgenii and Korobeinikov, Andrei:  Optimal quarantine strategies for COVID-19 control models
arXiv preprint arXiv:2004.10614 (2020).

\bibitem[Toda2020]{AT}  Alexis Toda: Susceptible-infected-recovered
(SIR) dynamics of Covid-19 and
economic impact
arXiv preprint:	arXiv:2003.11221 (2020).

\bibitem[Qun2 et al. 020]{QL}  Qun Li  et al.:  Early Transmission Dynamics in Wuhan, China, of Novel Coronavirus - Infected Pneumonia
The New England Journal of Medicine {\bf 382} (2020), 1199-1207.


\bibitem[Kai-Wang et al. 2020]{K-W}  Kai-Wang   et al.:  Temporal profiles of viral load in posterior oropharyngeal saliva samples and serum antibody responses during infection by SARS-CoV-2: an observational cohort study.
The Lancet Infectious Diseases {\bf 20} 5 (2020), 565-574.


\bibitem[Nishiura2020]{Nishiura2020}  Nishiura, H.:  Estimation of the asymptomatic ratio of novel coronavirus infections (COVID-19).
International Journal of Infectious Diseases {\bf 94}  (2020), 154-155.


\bibitem[Verity et al. 2020]{verity2020estimates}  Verity, Robert and Okell, Lucy C and Dorigatti, Ilaria and Winskill, Peter and Whittaker, Charles and Imai, Natsuko and Cuomo-Dannenburg, Gina and Thompson, Hayley and Walker, Patrick GT and Fu, Han    et al.:  Estimates of the severity of coronavirus disease 2019: a model-based analysis.
The Lancet Infectious Diseases {\bf 20} 6 (2020), 669-677.

\bibitem[Kirkpatrick83]{Kirkpatrick83}  Kirkpatrick, S. and Gelatt, C.~D. and Vecchi, M.~P.:  Optimization by Simulated Annealing.
Science {\bf 220} 4598 (1983), 671-680.

\bibitem[Jones2020]{jones2020optimal}  Jones, Callum and Philippon, Thomas and Venkateswaran, Venky:  Optimal Mitigation Policies in a Pandemic.
Working paper (2020).

\bibitem[Byrne et al. 2020]{Byrne}  Byrne et al.:  Inferred duration of infectious period of SARS-CoV-2: rapid scoping review and analysis of available evidence for asymptomatic and symptomatic COVID-19 cases.
medRxiv (2020).




 
 
\end{thebibliography}
\end{document}